\documentclass[a4paper,onecolumn,11pt, unpublished]{quantumarticle}

\pdfoutput=1
\usepackage[utf8]{inputenc}
\usepackage[english]{babel}
\usepackage[T1]{fontenc}

\usepackage{amsmath}
\usepackage[colorlinks=true, linkcolor=blue, citecolor=blue, urlcolor=blue, anchorcolor=black]{hyperref}

\usepackage{tikz}
\usepackage{lipsum}

\usepackage{amsthm}

\newtheorem{theorem}{Theorem}

\usepackage{comment}
\usepackage{amsfonts}
\usepackage{diagbox}
\usepackage{color}
\usepackage{dsfont}
\usepackage{graphicx}
\usepackage{subcaption}
\usepackage{amssymb}
\usepackage{amsmath,amssymb,bm,amsfonts,mathrsfs,bbm}
\usepackage{physics}

\newtheorem{proposition}{Proposition}
\newtheorem{corollary}{Corollary}
\newtheorem{definition}{Definition}
\newtheorem{lemma}{Lemma}

\usepackage[numbers]{natbib}

\usetikzlibrary{shapes}
\usepackage{tikz-3dplot}
\usetikzlibrary{intersections, backgrounds}
\usepackage{soul,xcolor}


\begin{document}

\title{Sufficient criteria for absolute separability in arbitrary dimensions via linear map inverses}

\author{Jofre Abellanet-Vidal$^\dag$*}
\email{jofre.abellanet@uab.cat}
%thanks{These authors contributed equally to this work.}
\affiliation{F\'isica Te\`orica: Informaci\'o i Fen\`omens Qu\`antics, Departament de F\'isica, Universitat Aut\`onoma de Barcelona, E-08193 Bellaterra, Spain}

\author{Guillem M\"uller-Rigat*}
\email{guillem.muller@icfo.eu}
%thanks{These authors contributed equally to this work.}
\affiliation{ICFO-Institut de Ci\`encies Fot\`oniques, The Barcelona Institute of Science and Technology, Castelldefels (Barcelona) 08860, Spain.}

\author{Grzegorz Rajchel-Mieldzioć}
\affiliation{BEIT sp. z o.o., ul.\ Mogilska 43, 31-545 Krak{\'o}w, Poland}
\affiliation{ICFO-Institut de Ci\`encies Fot\`oniques, The Barcelona Institute of Science and Technology, Castelldefels (Barcelona) 08860, Spain.}

\author{Anna Sanpera}
%\email{anna.sanpera@uab.cat}
\affiliation{F\'isica Te\`orica: Informaci\'o i Fen\`omens Qu\`antics, Departament de F\'isica, Universitat Aut\`onoma de Barcelona, E-08193 Bellaterra, Spain}
\affiliation{ICREA -- Instituci\'o Catalana de Recerca i Estudis Avan\c{c}ats, Lluis Companys 23, 08010 Barcelona, Spain}

\maketitle
\renewcommand\thefootnote{}
\footnotetext{*These authors contributed equally to this work.}
\footnotetext{$\dag$ Corresponding author.}
\renewcommand\thefootnote{\arabic{footnote}}

\begin{abstract}
Quantum states that remain separable (i.e., not entangled) under any global unitary transformation are known as \textit{absolutely separable} and form a convex set. Despite extensive efforts, the complete characterization of this set remains largely unknown. In this work, we employ linear maps and their inverses to derive new sufficient analytical conditions for absolute separability in arbitrary dimensions, providing extremal points of this set and improving its characterization. Additionally, we employ convex geometry optimization to refine the characterization of the set when multiple non-comparable criteria for absolute separability are available. We also address the closely related problem of characterizing the \textit{absolute PPT} (positive partial transposition) set, which consists of quantum states that remain positive under partial transposition across all unitary transformations. Finally, we extend our results to multipartite states. \textit{We are proud to dedicate our work to Professor Ryszard Horodecki, whose pioneering contributions to the field of quantum entanglement continue to inspire us all. With deep gratitude and respect.}
\end{abstract}

\section{Introduction}
Despite decades of intense research, necessary and sufficient conditions to ensure separability exist only in the simplest cases~\cite{peres_separability_1996, horodecki_separability_1996} or for highly restricted families of states.
In fact, entanglement verification and quantification is known to be generically an NP-hard problem~\cite{gharibian_strong_2010, gurvits_classical_2003}. 
One such necessary and sufficient condition involves positive but not completely positive maps; however, to date, neither a finite set of such maps~\cite{skowronek_there_2016} nor semidefinite programming methods~\cite{doherty_complete_2004, guhne_entanglement_2009} can definitively determine whether a given state is entangled~\cite{horodecki_five_2022, knill_iqoqi_2013}. Given these limitations, rather than seeking \textit{iff} conditions, one may instead focus on deriving sufficient conditions for separability.

 Here, we draw our attention to the use of the inverses of positive but not completely positive maps as a sufficient separability criterion~\cite{lewenstein_sufficient_2016, filippov_absolutely_2017, lewenstein_linear_2022}. 
 In particular, we examine \textit{absolute separability} (AS), which refers to those states that remain separable under all unitary transformations~\cite{kus_geometry_2001}. Characterizing the set of AS is crucial for various reasons. Beyond its significance in fully describing the set of quantum states, AS states play a key role in the resource theory of entanglement, where they represent free states, while entangled states serve as a resource \cite{chitambar_quantum_2019, patra_resource_2023}. Since unitary transformations are the primary way of generating entanglement using e.g., quantum circuits or quenches, a state certified as AS cannot become entangled through unitary transformations~\cite{halder_characterizing_2021}. Moreover, the characterization of such states must depend only on 
their spectrum (eigenvalues), and their detection does not require full quantum state tomography \cite{ekert_direct_2002,tanaka_determining_2014}. For a quantum state of dimension $D=N\cdot M$ (with $N,M$ representing a bipartition), full tomography requires reconstructing $D^{2}-1$ real parameters, while only $D-1$ eigenvalues are needed to characterize the spectrum. Finally, since the AS set is convex, its characterization relies on identifying its extremal points.

 Several sufficient criteria for absolute separability (AS) have been proposed over the years. These include bounding matrix norms~\cite{gurvits_largest_2002, hildebrand_entangled_2007}, using linear matrix inequalities~\cite{hildebrand_ppt_2007, champagne_spectral_2022}, establishing the equivalence of absolute separability and absolute PPT for certain subsets~\cite{verstraete_maximally_2001, johnston_separability_2013}, employing separability witnesses~\cite{serrano-ensastiga_absolute-separability_2024, patra_efficient_2021}, computing moments \cite{mallick_detection_2025}, considering positive reduction from the spectrum \cite{jivulescu_positive_2015} and identifying extremal points~\cite{song_extreme_2024}, among others. Despite these advancements, the complete characterization of the AS set remains an open problem~\cite{arunachalam_is_2015}.

Here, we present a more refined characterization of the AS set for arbitrary dimensions. Our findings yield compact analytical conditions for absolute separability that depend on only a few, or even a single eigenvalue of the state. Using similar techniques, we also derive bounds on the AS set in multipartite scenarios, applicable to both qubits and qudits. Due to the convexity of the AS set, criteria for AS can be expressed as convex conditions on the state $\rho$. 
Here, we show how to merge previous criteria of AS with our own ones into a unifying framework based on convex programming, a class of efficiently solvable optimization algorithms. The separability tests we propose not only reconcile various separability conditions but also enable systematic improvements over them. In doing so, we offer a new perspective on the long-standing open problem of \textit{separability from the spectrum} in quantum information theory~\cite{knill_iqoqi_2013}.

The paper is organized as follows: In Section~\ref{sec:AS}, we introduce the formal definitions of absolute separability (AS) and its twin problem, absolute positive partial transposition (AP). We review previous criteria and derive new theorems for AS by inverting non-completely positive maps, which significantly improve upon previous results. In Section~\ref{sec:CH}, we present new bounds on absolute separability using convex geometric tools, thereby unifying previously incomparable criteria. Notably, our methods are applicable to arbitrary dimensions. In Section~\ref{sec:multi}, we explore absolute separability in multipartite systems, where we derive analytical conditions based on a single eigenvalue. Finally, we focus on the symmetric subspace, where we establish membership conditions for the AS and AP sets. We conclude in Section~\ref{sec:conclus} and propose several open conjectures stemming from our results.

\section{Absolute separability}
\label{sec:AS}
Let us recall that a quantum state is separable if it can be represented as a convex combination of product states. For bipartite states, $\rho\in \mathcal B(\mathbb{C}^{N}\otimes\mathbb{C}^{M})$, where $\mathcal B(\mathbb{C}^{N} )$ denotes the set of bounded operators acting on a Hilbert space
$\mathcal H= \mathbb{C}^{N}$,  that implies
$$\rho=\sum_i p_i \ket{e_i,f_i}\bra{e_i,f_i},$$
with $\ket{e_i}\in \mathbb{C}^{N}$, $\ket{f_i}\in\mathbb{C}^{M}$, $p_i\geq 0$ and $\sum_i p_i=1$. 
Notice that all entangled pure states
can be obtained by applying a certain global unitary transformation $U$ to a given product state, i.e., $\ket{\Psi}_{NM}=U(\ket{e}_N\ket{f}_M)$, however, this no longer holds for mixed states. 
Formally, we define absolute separability as:

\begin{definition}
    A state $\rho \in \mathcal B(\mathbb{C}^{N}\otimes\mathbb{C}^{M})$ is absolutely separable (AS) if and only if, for any unitary matrix $U$ acting on $ \mathbb{C}^{N}\otimes\mathbb{C}^{M}$, $\rho'=U\rho U^{\dag}$ is separable. 
 \end{definition}

Since unitary transformations preserve the spectrum of a given density matrix $\rho$, it is possible to infer \textit{absolute} separability criteria based solely on the eigenvalues of the state \cite{verstraete_maximally_2001}.  For that reason, the set AS is also called separable from spectrum. Clearly, the maximally mixed state (MMS),  $\rho\propto\mathds{1}$, is absolutely separable. Moreover, states sufficiently close to the MMS are also absolutely separable, as discussed in the following. \\

The most celebrated AS criterion is based on the matrix norm of the operator $\rho$, and is known as the \textit{Gurvits-Barnum ball} (GBB)~\cite{gurvits_largest_2002}. It can be expressed as a function of the purity of the state, $\mathrm{Tr}(\rho^2)$, as follows:

\begin{theorem}
\label{thm:Gurvits}\cite{gurvits_largest_2002}
Let  $\rho \in \mathcal B(\mathbb{C}^{N}\otimes\mathbb{C}^{M})$  be a normalized state in a Hilbert space of global dimension $D=N\cdot M$. If
\begin{equation}
\label{eq:GurvitsBi}
  \mathrm{Tr}(\rho^2)\leq \frac{1}{D-1},
\end{equation}
then $\rho$ is absolutely separable.
\end{theorem}

From the spectral decomposition of the state, $\rho = \sum_{i=0}^{D-1}\lambda_i\ketbra{\psi_i}$, condition Eq.~\eqref{eq:GurvitsBi} is geometrically interpreted as a ball of radius $(\sqrt{D-1})^{-1}$ in the space of eigenvalues $\boldsymbol{\lambda} = \{\lambda_i\}$. 
For normalized states, i.e.,
$\mathrm{Tr}(\rho) = \sum_{i=0}^{D-1}\lambda_i = 1$, this ball is centered on the maximally mixed state.\\

\noindent \textbf{Absolute separability from linear maps.--} Here, we take a different approach based on linear maps to tackle the characterization of AS.
Recall that positive but not completely positive linear maps, $\Lambda:\mathcal B(\mathcal H)\rightarrow \mathcal B(\mathcal H')$,  provide \textit{sufficient} conditions for entanglement (or, equivalently, \textit{necessary} conditions for separability) when applied to just one of the subsystems of a bipartite state
\cite{horodecki_separability_1996, terhal_family_2001, lewenstein_characterization_2001}. 
Interestingly, as reported in~\cite{lewenstein_sufficient_2016, lewenstein_linear_2022}, such maps are also instrumental to derive different \textit{sufficient} conditions for separability. Before proceeding further, we present the central result on which our work relies: 

\begin{theorem}
\label{thm:Primer}\cite{lewenstein_sufficient_2016}
    Let S, S' be convex and compact subsets of the space of all bounded operators $\mathcal{B}(\mathbb{C}^{N}\otimes\mathbb{C}^{M})$, and let $\Lambda_{\mathbf{p}}:S\rightarrow S'$ be a family of maps, invertible for almost all $\mathbf{p}$. By $\mathcal{P}_{SS'}$ we denote the subset of the parameters set $\mathbf{p}$, for which the maps have the property that, for every $\rho\in S$, $\Lambda_{\mathbf{p}}(\rho)\in S'$ provided $\mathbf{p}\in\mathcal{P}_{SS'}$. Then, if $\Lambda_{\mathbf{p}}^{-1}(\sigma)\in S$, it follows that $\sigma \in S'$. 
\end{theorem}

In  particular , if the set  S  corresponds to all quantum states acting on $\mathbb{C}^{N}\otimes\mathbb{C}^{M}$, and the set $S'$ is the subset of separable states, the above theorem provides a sufficient criterion for certifying that a given state $\sigma\in \mathcal{B}(\mathbb{C}^{N}\otimes\mathbb{C}^{M}$) is separable. As pointed out in \cite{lewenstein_sufficient_2016}, the choice of $\Lambda_\mathbf{p}$ is determined first by the capability of proving the assumption that $\Lambda_{\bf p}(S)\subset S'$, and second by checking that $\Lambda_{\bf p}^{-1}(S')\subset S$, that is, by the difficulty of inverting analytically the chosen map. \\

The previous  theorem allows  us to derive sufficient criteria for AS. To this end, we define $S$  as the set of all quantum states, i.e.,  
$\rho\geq 0$ (or, up to normalization, positive semi-definite (PSD) matrices), and $S'$ as the subset of states that are separable from their spectrum. Within this framework, it is natural to consider unitarily covariant linear maps 
$\Lambda$, i.e., those fulfilling  that for all $\rho$ and $U$ acting on $\mathbb{C}^{N}\otimes\mathbb{C}^{M}$, $\Lambda(U \rho U^{\dag})=V(U)\Lambda(\rho)V(U)^{\dag}$, where $V(U)$ is also unitary \cite{bardet_characterization_2020}.  
As shown in \cite{bhat_linear_2011}, the only linear maps preserving unitary equivariance on the map and its inverse are the so-called reduction-like maps \cite{cerf_reduction_1999}, which are parameterized as:
\begin{equation}
    \label{eq:ReductionMap}
    \Lambda_{\alpha}(\rho)=\Tr(\rho)\cdot \mathds{1}+\alpha\cdot\rho,
\end{equation}
with the single parameter $\alpha\in \mathbb{R}$. This family of linear maps 
is invertible for any nonzero $\alpha$. Its inverse, which is linear, reads:  
\begin{equation}
    \label{Inversereduction}
    \Lambda_{\alpha}^{-1}(\sigma)=\frac{1}{\alpha}\left(\sigma-\frac{\Tr(\sigma)\cdot\mathds{1}}{D+\alpha}\right),
\end{equation}
where $D$ is the dimension of the Hilbert space in which $\sigma$ acts.

As demonstrated in \cite{lewenstein_sufficient_2016}, 
the map,  $\Lambda_{\alpha}(\rho)$, 
transforms any state $\rho$ into a separable state for $\alpha\in[-1,2]$. Consequently, Theorem~\ref{thm:Primer} establishes that the positivity of the inverted map, $\Lambda_{\alpha}^{-1}(\sigma)$
serves as a \textit{sufficient} condition for separability. Since the positivity of $\Lambda_{\alpha}^{-1}(\sigma)$ depends only on the spectrum of the state $\sigma$, and the map is unitarily invariant, this condition directly certifies AS. Furthermore, because the conditions for positivity are bounded by  $\alpha_-=-1$, and $\alpha_+=2$, these values provide sufficient criteria for AS based on a single eigenvalue, either the minimum or the maximum, as stated below.
\begin{theorem}
\label{Prop:Lewenstein}
Let $\rho \in \mathcal B(\mathbb{C}^{N}\otimes\mathbb{C}^{M})$  be a normalized state in a Hilbert space of global dimension $D=N\cdot M$,
with minimum and maximum eigenvalues $\lambda_{\min}(\rho)$ and $\lambda_{\max}(\rho)$ respectively. If
\begin{equation}
    \lambda_{\min}(\rho) \geq \frac{1}{D +2}
    \quad \text{or} \quad
    \lambda_{\max}(\rho) \leq \frac{1}{D-1},
\label{eq:redcond}
\end{equation}
 then $\rho$ is absolutely separable. 
\end{theorem}
\begin{proof}
(See also \cite{lewenstein_sufficient_2016, vidal_robustness_1999}). 
For $2\geq \alpha > 0$, positivity of $\Lambda_{\alpha}^{-1}(\rho)$ demands that $\rho - \mathds{1}/(N\cdot M+\alpha)\geq 0$, that is, $\lambda_{\min}(\rho) - 1/(M\cdot N +\alpha)\geq 0$. 
On the other hand, for $0>\alpha\geq -1$, the condition reads $\mathds{1}/(M\cdot N+\alpha) - \rho\geq 0 $, or equivalently, $1/(M\cdot N+\alpha)- \lambda_{\max}(\rho)\geq 0$.  The extreme values $\alpha = 2$ and $\alpha = -1$ yield the announced result.
\end{proof}
Finally, we note that both conditions contained in Theorem~\ref{Prop:Lewenstein} and the conditions on Gurvits-Barnum ball, Theorem \ref{thm:Gurvits}, depend solely on the dimension of the global bipartite system $D=N\cdot M$, and are independent of the specific local dimensions $N$ and $M$. For a multipartite state of global dimension $D$, these conditions provide sufficient criteria for absolute separability with respect to \textit{any} bipartition of the total system (see Section~\ref{sec:multi}). Moreover, as we will show later, the conditions derived from Theorem~\ref{Prop:Lewenstein} define extremal points in the convex set of AS states, enabling a geometric approach to its characterization. The classification of the extreme points of AS states is only known for the specific case of qubit-qudit systems \cite{johnston_separability_2013}.

Since all separable states are necessarily positive under partial transposition (PPT), it is natural to ask whether there exists a set of states that remain PPT under any unitary transformation, and whether this set coincides with the AS set \cite{hildebrand_ppt_2007, arunachalam_is_2015}. By analogy to the AS set, we define absolute PPT (AP) as:

\begin{definition}
A normalized state $\rho \in \mathcal B(\mathbb{C}^{N}\otimes\mathbb{C}^{M})$ is absolutely PPT 
(AP) if and only if, for any unitary matrix $U$ acting on $\mathbb{C}^{N}\otimes\mathbb{C}^{M}$, $\rho'=U\rho U^{\dag}$ is PPT.
\end{definition}

Unlike the AS set, the AP set has already been fully characterized for arbitrary dimensions through a series of linear matrix inequalities (LMIs), originally derived in~\cite{hildebrand_ppt_2007}. Despite the number of LMIs required to certify AP grows exponentially with $\min\{N,M\}$, making its characterization infeasible in high dimensions, a necessary AP condition can be obtained from the positivity of the first $2\times 2$ block of the first LMI (see Theorem~1 of~\cite{hildebrand_ppt_2007}). Considering the eigenvalues of a state $\rho\in\mathcal{B}(\mathbb{C}^{N}\otimes\mathbb{C}^{M})$ in non-decreasing order, i.e., $\lambda_{0}\leq\lambda_{1}\leq\cdots\leq\lambda_{N\cdot M-1}$, the necessary condition for AP reads:
\begin{equation}
\label{eq:2x2LMI}
\begin{pmatrix}
    2\cdot\lambda_{0} & \lambda_{1}-\lambda_{NM-1}\\
    \lambda_{1}-\lambda_{NM-1} & 2\cdot\lambda_{2}
\end{pmatrix}    \geq 0,
\end{equation}
involving only the three lowest eigenvalues, 
$\{ \lambda_{0}\leq\lambda_{1}\leq \lambda_{2} \}$, 
plus the largest one 
$\lambda_{N\cdot M-1}$. \\

Interestingly, the above condition is satisfied only by states whose minimum and maximum eigenvalues fulfill one of the conditions of Theorem~\ref{Prop:Lewenstein}. That is, those detected with the inverse of the reduction map for parameters $\alpha\in[-1,2]$. Consequently, the region of AP identified by the bounds of $\alpha$ derived from Eq.\eqref{eq:2x2LMI} coincides with the region of AS detected  by 
Eq.\eqref{eq:redcond}. Although Theorem~\ref{Prop:Lewenstein} does not provide a complete characterization of the AS set, it holds for any local dimensions and in this way, supports the conjecture that AS and AP are identical sets. Only for the qubit-qudit case, $(\mathbb C^2\otimes \mathbb C^d)$, it has been proven that AP is also a sufficient condition for AS~\cite{johnston_separability_2013}. However, for arbitrary bipartition dimensions, the equivalence between the AS and AP sets remains an open conjecture \cite{arunachalam_is_2015}.
 
\section{Improving sufficient criteria for absolute separability in arbitrary dimensions from convex geometry}
\label{sec:CH}

As mentioned previously, the sets of states detected by the inverse of the reduction map, and those enclosed within the Gurvits-Barnum ball are not directly comparable. However, since the AS set is convex and compact \cite{ganguly_witness_2014}, we can unify both criteria using convex geometrical tools into a single convex hull that yields stronger separability conditions.

\subsection{Extremal points of  the  convex hull of AS criteria }

The set of AS states detected through Theorem \ref{Prop:Lewenstein}, 
form simplexes in the space of eigenvalues. For arbitrary dimensions, $D=N\cdot M$, the vertices of these simplexes correspond to all permutations on the vector of eigenvalues. Explicitly, for $\alpha=2$, the largest simplex corresponds to the vector $\left(\frac{1}{D+2}, \cdots, \frac{1}{D+2}, 1-\frac{D-1}{D+2}\right)$ and its permutations, while for  $\alpha=-1$, the vertices are settled by  $\left(0, \frac{1}{D-1},\cdots, \frac{1}{D-1}\right)$ and its permutations.
A schematic representation of these sets for $D=3$ is depicted in Figure~\ref{fig:Simplex} using barycentric coordinates $\{\tilde{\lambda}_i\}$. In this representation, the maximally mixed state (MMS) lies at the origin and pure states correspond to the vertices of the simplex of normalized states.
\begin{figure}[htbp!]
\centering
\begin{tikzpicture}[scale=3.5]

\fill[black, opacity=0.15] 
    (0,0.81649658) -- 
    (-0.35355,0.204) --
    (0,0) -- 
    cycle;

\fill[black, opacity=0.3] 
        (0,0.326598) -- 
        (-0.35355,0.204) --
        (0,0) -- 
        cycle;

    \draw[blue, thick, solid] (0,0.81649658) -- (-0.70710678,-0.40824829) -- (0.70710678,-0.40824829) -- cycle;
    \draw[green, very thick] (0,0.326598) -- (0.3536,0.204) -- (0.2828,-0.1633) -- (0,-0.408) -- (-0.2828,-0.163299) -- (-0.35355,0.204) ;
    \draw[green, very thick] (0,0.326598) -- (-0.35355,0.204) ;
    \draw[orange, very thick, dotted] (0,-0.408) -- (-0.35355,0.204) -- (0.3536,0.204) -- cycle;
    \draw[red, very thick, dashed] (0,0.326598) -- (-0.2828,-0.163299) -- (0.2828,-0.1633) -- cycle;

\node at (0,0.326598) [cross out, draw=red, very thick, scale=0.6] {};
\node at (-0.2828,-0.163299) [cross out, very thick, draw=red, scale=0.6] {};
\node at (0.2828,-0.1633) [cross out, very thick, draw=red, scale=0.6] {};

\node at (0,-0.408) [diamond, fill=orange, scale=0.5] {};
\node at (0.3536,0.204) [diamond, fill=orange, scale=0.5] {};
\node at (-0.35355,0.204) [diamond, fill=orange, scale=0.5] {};
    
    \filldraw[fill=blue,draw=blue] (0,0.81649658) circle (0.7pt);
    \filldraw[fill=blue,draw=blue] (-0.70710678,-0.40824829) circle (0.7pt);
    \filldraw[fill=blue,draw=blue] (0.70710678,-0.40824829) circle (0.7pt);

    %\draw[->] (-0.8,0) -- (0.8,0) node[right] {$\tilde{\lambda}_{1}$};
    %\draw[->] (0,-0.5) -- (0,1) node[above] {$\tilde{\lambda}_{2}$};

\draw[->] (0,-0.408) -- (0,0.9) node[right] {$\lambda_0$};
\draw[->] (0.3536,0.204) -- (-0.77,-0.445) node[above left] {$\lambda_1$};
\draw[->] (-0.35355,0.204) -- (0.77,-0.445) node[above right] {$\lambda_2$};

\end{tikzpicture}
    \caption{Schematic representation of quantum states and the AS set for $D=3$ in barycentric coordinates. Normalized states fill the area enclosed by the blue vertices. In red (dashed) the simplex corresponding to $\lambda_{\min}\geq 1/(D+2)$, in orange(dots) the one  corresponding to  $\lambda_{\max}\leq 1/(D-1)$ (see Theorem~\ref{Prop:Lewenstein}).  In green the convex hull of both simplexes. The light shaded region fulfills the ordering  $\lambda_{0}\leq\lambda_{1}\leq\lambda_{2}$ we consider.
    The dark shadow polytope highlights the intersection of the convex hull with the ordered zone. The figure is illustrative as $D=3$ does not correspond to any bipartite splitting.}
    \label{fig:Simplex}
\end{figure}
The displayed simplexes shown for AS are, in fact, extreme due to the tightness of the values of $\alpha$ obtained from Eq.~\eqref{Inversereduction}. Thus, the convex hull arising from the union of both simplexes is optimal. 
The following Lemma fully describes, for given values of $\alpha_{\pm}$, the convex hull of the two simplexes. 

\begin{lemma}
\label{Lemma:GeneralAlphaPM} 
Let $\rho \in \mathcal B(\mathbb{C}^{N}\otimes\mathbb{C}^{M})$  be a normalized state in a Hilbert space of global dimension $D=N\cdot M$, and $\boldsymbol{\lambda} = \{\lambda_i \}_{i=0}^{D-1}$ the corresponding eigenvalues in a non-decreasing order, that is, $0\leq\lambda_0\leq\lambda_1 \cdots \leq\lambda_{D-1}$, with $\sum_{i=0}^{D-1}\lambda_i = 1$.  Given $\alpha_{+}\geq 0$ and $\alpha_{-}\leq0$, the vector $\boldsymbol{\lambda}$ is contained in the convex hull settled by conditions 

 \begin{equation}
\label{eq:DemoLemma1}
\lambda_{0}\geq \frac{1}{D+\alpha_{+}},  
\end{equation}
\begin{equation}
\label{eq:DemoLemma2}
\lambda_{D-1}\leq \frac{1}{D+\alpha_{-}},   
\end{equation} 
if and only if
 \begin{equation}
 \label{eq:GeneralAlphaPM}
K \cdot \sum_{i=0}^{c-1} \lambda_i + \left[D -K \cdot c + \alpha_+ \right] \cdot \lambda_c \geq 1,
\end{equation}
where $ K= \left( 1 - \frac{\alpha_+}{\alpha_-} \right)$, and $c = \left\lceil \frac{\alpha_+ + \alpha_- (D - 1 + \alpha_+)}{\alpha_- - \alpha_+} \right\rceil$.
\end{lemma}
\begin{proof} 
In what it follows we demonstrate how the single inequality given by Eq.~\eqref{eq:GeneralAlphaPM}, is indeed necessary and sufficient to describe the convex hull given by conditions Eq.~\eqref{eq:DemoLemma1} and \eqref{eq:DemoLemma2}, for any dimension $D$. See also Appendix~\ref{AppendixDerivationSimplex} for a numerical study.

The \textit{if} part of the proof follows from an exhaustive analytical characterization of the extremal points defined by the constraints Eq.~\eqref{eq:DemoLemma1},\eqref{eq:DemoLemma2}, together with the $D-1$ ordering conditions $\{ \lambda_{k} \leq \lambda_{k+1}\}_{k=0,1,..,D-2}$ and the PSD constraint $\lambda_0\geq 0$. The feasible set of eigenvalues that satisfy these conditions forms a polytope in the $\mathbb{R}^{D-1}$ affine subspace defined by the normalization condition $\sum_{i=0}^{D-1}\lambda_i = 1$ (see the dark  shadow  polytope in Figure~\ref{fig:Simplex}). In the dual picture, such a polytope can be equivalently described as a convex combination of extremal points (vertices). Each vertex is necessarily specified by the saturation of $D-1$ inequalities and the fulfillment of the remaining ones. The proof relies on the (analytical) computation of all such points and the verification that they are detected either by Eq.~\eqref{eq:DemoLemma1} or Eq.~\eqref{eq:DemoLemma2}.

To begin with the characterization of such vertex, we first consider the case  \( \{ \lambda_k =  \lambda_{k+1} \}_{k=0,1,\dots,D-2} \).
Under normalization, it leads to $\boldsymbol{\lambda}_{\rm MMS} = \left(1/D,1/D,...,1/D\right)$, corresponding to the MMS, which is detected by both conditions. The second option is to consider Eq.~\eqref{eq:GeneralAlphaPM} and pick $D-2$ conditions (i.e., all except one) from $\{ \lambda_k = \lambda_{k+1}\}_{k=0,1,..,D-2}$. Removing the $r+1$-th condition from the last set implies that the vector is of the form $\boldsymbol{\lambda}^{(r)} = (\underbrace{a,...,a}_{r},\underbrace{b,...,b}_{D-r})$ with $a\leq b$. Notice that geometrically, $\boldsymbol{\lambda}^{(r)}$ generate a cone whose apex corresponds to the MMS which is enclosed by the $D-1$ eigenvalue ordering conditions.The normalization condition $r\cdot a + (D-r)\cdot b = 1 $ constraints these two quantities. Let us find $a\leq 1/D$ and $b\geq 1/D$ such that the proposed inequality Eq.~\eqref{eq:GeneralAlphaPM} is also saturated by $\boldsymbol{\lambda}^{(r)}$. For $r\leq c$, we obtain $a = (r+\alpha_-)/[r(D+\alpha_-)],b = 1/(D+\alpha_-)$, which is detected by Eq.~\eqref{eq:DemoLemma2}. On the other hand, for $r>c$, we obtain $a = 1/(D+\alpha_+), b = (D-r+\alpha_+)/[(D-r)(D+\alpha_+)]$, which is detected by Eq.~\eqref{eq:DemoLemma1}. The remaining cases involve $\lambda_0 = 0$ (i.e., rank-deficient states). The only solution here is for $\alpha_- = -1$, 
$r=1, a= b= 1/(D-1)$, since considering $r> 1$ or $0\geq \alpha_->-1$ leads to incompatible equations for all $D$. 

The \textit{only if} part is a consequence of the fact that condition Eq.~\eqref{eq:GeneralAlphaPM} is Schur-convex, and it is saturated by the extremals of both simplices, i.e. $\boldsymbol{\lambda}_{\pm} = \left(\frac{1}{D+\alpha_{\pm}}, \cdots, \frac{1}{D+\alpha_{\pm}}, 1-\frac{D-1}{D+\alpha_{\pm}}\right)$. Further elaboration on this claim can be found in the Appendix~\ref{Appendix:Ansatz}.
\end{proof}

The previous lemma, together with the convexity of the AS set, enables to derive new sufficient conditions for AS, which are strictly stronger than those proposed in Theorem~\ref{Prop:Lewenstein}. 

\begin{theorem}
\label{Theorem2:LinearIneqaulityBipartite}
 Let $\rho$ be a normalized bipartite state acting in a Hilbert space 
 $\mathbb{C}^N\otimes\mathbb{C}^{M}$ of dimension $D=N\cdot M$, and let $\boldsymbol{\lambda} = \{\lambda_i \}_{i=0}^{D-1}$ denote its corresponding eigenvalues in non-decreasing order, i.e.,  $0\leq \lambda_0\leq\lambda_1\cdots\ \leq \lambda_{D-1}$ with $\sum_{i=0}^{D-1}\lambda_i = 1$. If
\begin{equation}
    \label{eq:TwoSimplex}
     3\cdot \sum_{i=0}^{\lceil \frac{D-1}{3}\rceil-1}\lambda_{i}+\left( D+2-3\cdot\left\lceil\frac{D-1}{3}\right\rceil \right)\cdot \lambda_{\lceil\frac{D-1}{3}\rceil}\geq 1,
\end{equation}
then $\rho$ is AS. 
\end{theorem}
\begin{proof}
The proof follows from Lemma~\ref{Lemma:GeneralAlphaPM} by considering $\alpha_{+}=2$ and $\alpha_{-}=-1$ (Eq.~\eqref{eq:redcond}).
\end{proof}
Let us summarize our findings and illustrate them: the single inequality given by Eq.~\eqref{eq:TwoSimplex} is satisfied by all states detected by $\lambda_{0}\geq 1/(D+2)$ and $\lambda_{D-1}\leq 1/(D-1)$. Moreover, it is also satisfied by states not detected by the previous conditions but which belong to the convex hull of both simplexes.  
For example, consider a bipartite $2$-qubit state, with total dimension $D=4$, whose eigenvalues are given by $\boldsymbol{\lambda}=(1/12, 1/4, 1/4, 5/12)$. This state cannot be certified AS through the conditions of Theorem \ref{Prop:Lewenstein} which correspond to $\lambda_0\geq 1/6$ nor by $ \lambda_{D-1}\leq 1/3$, yet the state is detected by Eq.~\eqref{eq:TwoSimplex}. This ensures the existence of at least one convex decomposition onto states that are detected either from $\lambda_0$ or 
$\lambda_{D-1}$. Indeed, such decomposition is given by the equally weighted contribution of the vertex of each simplex $\boldsymbol{\lambda}=\frac{1}{2}\left[(1/6,1/6,1/6, 1/2) + (0, 1/3, 1/3, 1/3) \right]$.\\

The inequality given by Eq. \eqref{eq:TwoSimplex} involve several eigenvalues, however, it can be relaxed to a hierarchy set of weaker conditions for AS,  containing fewer eigenvalues. Given the  non-decreasing  ordering, $\{\lambda_{k}\leq\lambda_{k+1}\}$, one can replace $\lambda_{k+1}$ by $\lambda_{k}$ in Eq.~\eqref{eq:TwoSimplex}, to yield new sufficient conditions for AS. Albeit being weaker, such conditions are based on fewer eigenvalues than the original theorem, which enables the detection of AS from incomplete knowledge of the spectrum. By iteratively applying the above substitution for $i\leq \kappa\leq \lceil \frac{D-1}{3}\rceil$, we derive the following hierarchy of AS conditions:

\begin{equation}
\begin{aligned}
    \kappa &= 0:  \quad 3\cdot\lambda_{0} + \cdots + 3\cdot\lambda_{\lceil\frac{D-1}{3}\rceil-1}+(D+2-3\cdot\left\lceil \frac{D-1}{3}\right\rceil)\cdot\lambda_{\lceil\frac{D-1}{3}\rceil}  \geq 1, \\
    \kappa &= 1:  \quad 3\cdot\lambda_{0} + \cdots + (D+5-3\cdot\left\lceil \frac{D-1}{3}\right\rceil)\cdot\lambda_{\lceil\frac{D-1}{3}\rceil-1}+ \geq 1, \\
     & \vdots \\
     \kappa &= \left\lceil\frac{D-1}{3}\right\rceil:  \quad
    (D+2) \cdot \lambda_{0}  \geq 1.
\end{aligned}
\label{eq:ReducedEigenvalues}
\end{equation}
Each $\kappa$-level in Eq.~\eqref{eq:ReducedEigenvalues} is weaker than the previous  $\kappa-1$, but does not require the knowledge of eigenvalue $\lambda_{\kappa}$, allowing also  for some intermediate eigenvalues to be unknown. By considering the two smallest eigenvalues, we derive the weaker condition $(D-1)\lambda_{1}+3\lambda_{0}\geq 1$, which already includes the $2D$ vertices of both simplexes despite being suboptimal. The whole  polytope of AS states is enclosed by all the possible permutations on the eigenvalues of inequality Eq.~\eqref{eq:TwoSimplex}, but just as with the case of the simplexes in Figure~\ref{fig:Simplex}, one linear inequality is enough for ordered eigenvalues (see Lemma~\ref{Lemma:GeneralAlphaPM}). 

\subsection{Extending Gurvits-Barnum criterion for arbitrary dimensions}
Let us observe that condition $\lambda_{\max}(\rho)\leq 1/(N\cdot M-1)$ is completely enclosed within the Gurvits-Barnum ball. Thus, we focus on the set of states detected by the convex hull (CH) of the ball (Theorem~\ref{thm:Gurvits}) plus the condition $\lambda_{\min}(\rho) \geq 1/(N\cdot M+2)$ (Theorem~\ref{Prop:Lewenstein}). This approach improves the existing separability criteria for the general dimension $D=N\cdot M$. Notably, the Hilbert-Schmidt distance between the vertices of the simplex and the maximally mixed state tends towards twice the radius of the ball as $N\cdot M\rightarrow\infty$. Figure~\ref{fig:Trvslmin} provides a 2D-sketch of the geometry of such a convex hull and how different regions are detected. 
\begin{figure}[ht!]
    \centering
    \begin{subfigure}[b]{0.38\textwidth}
    \vspace{-0.1\textwidth}
    \includegraphics[width=\textwidth]{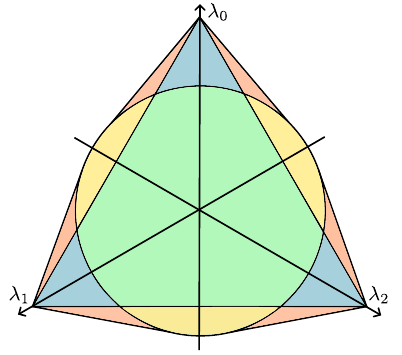}
    \vspace{0.1\textwidth}
    \label{fig:2b}
    \end{subfigure}
       \hspace{0.05\textwidth} 
       \vspace{0.01\textwidth} 
    \begin{subfigure}[b]{0.45\textwidth}
    \includegraphics[width=\textwidth]
    {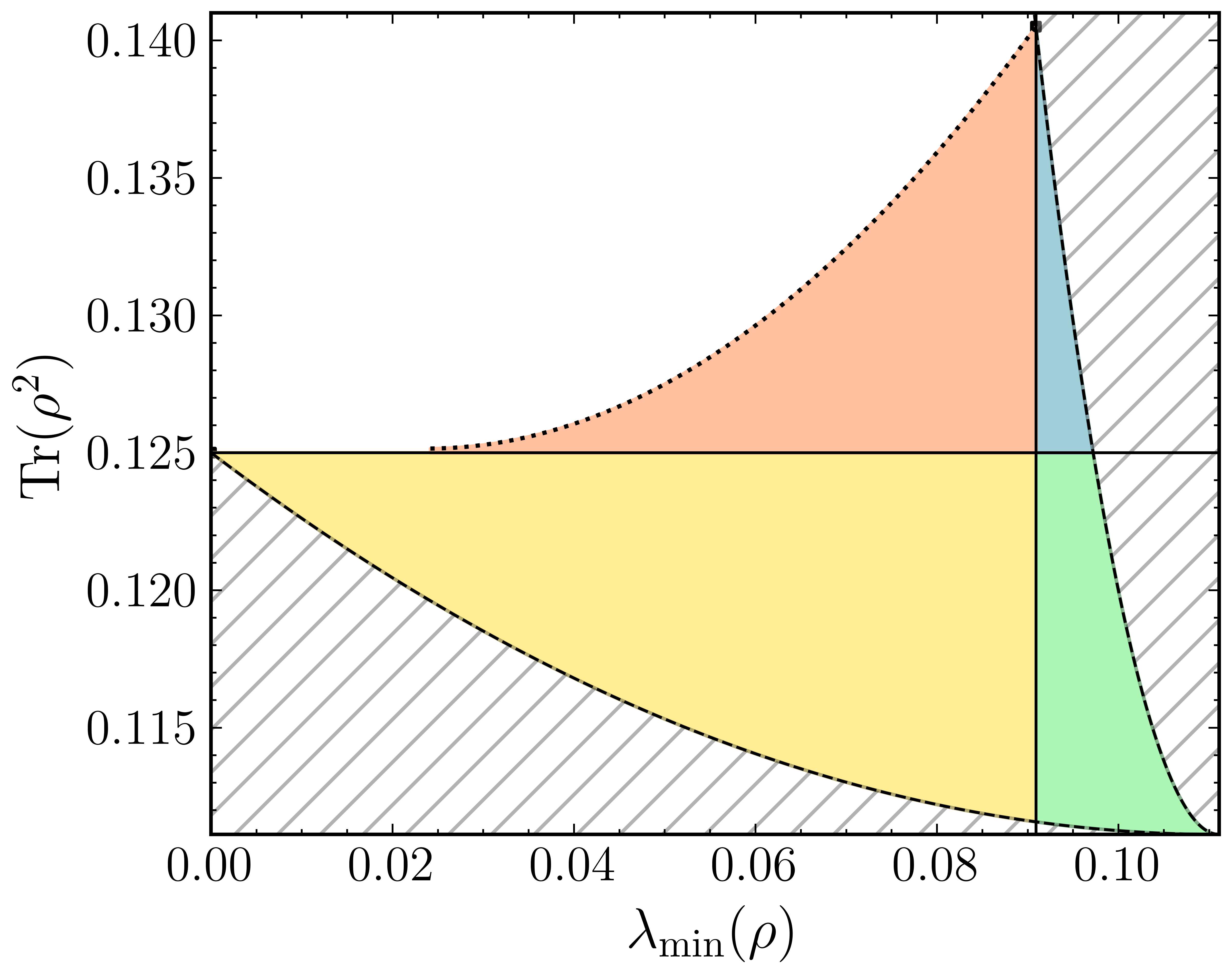}
    \label{fig:2a}
    \end{subfigure}
    \hspace{0.01\textwidth}
     \vspace{0.01\textwidth}

    \vspace{-0.05\textwidth}
    \caption{Detected absolutely separable (AS) states.
    \textbf{Left panel:} Schematic two-dimensional representation of the convex hull (orange) of the Gurvits-Barnum ball (yellow) and our simplex (blue) in barycentric coordinates (cf. Figure~\ref{fig:Simplex}).
    \textbf{Right panel:} Purity $\mathrm{Tr}(\rho^2)$ and minimal eigenvalue $\lambda_{\min}(\rho)$ of the detected states for a $2$-qutrit system. The points in the striped region do not correspond to physical states.}
    \label{fig:Trvslmin}
\end{figure}
While the analytical expression for the convex hull (Figure~\ref{fig:Trvslmin}) 
in high dimensions is long and challenging to work with, it is possible to express the problem as a disciplined convex program (DCP), as formalized in the next proposition.

\begin{proposition}
\label{Prop:SDP}
Let $\rho$ be a bipartite state acting in a Hilbert space $\mathbb{C}^{N}\otimes\mathbb{C}^{M}$ of global dimension $D=N\cdot M$. If the problem 
    \begin{equation}
\label{eq:SDPMapsNorm}
    \begin{array}{crl}
    \min_{\{\varrho_1\geq 0, \varrho_2\geq 0\}}& 0&\\
    \mbox{s.t.}& \rho& = \varrho_1 + \varrho_2\\
    &\lambda_{\min}(\varrho_1)&\geq \frac{\mathrm{Tr}(\varrho_1)}{D+2} \\
    & \mathrm{Tr}(\varrho^2_2) &\leq \frac{\mathrm{Tr}(\varrho_2)^2}{D-1}
    \end{array}
\end{equation}
is feasible, then $\rho$ is absolutely separable (AS).
\end{proposition}
\begin{proof}If the problem is feasible, it implies that the algorithm indeed succeeded in finding $\{\varrho_1, \varrho_2\}$ fulfilling the constraints.  Hence, both states $\{\varrho_1, \varrho_2\}$ are certified AS from the positivity of the inverse map introduced in Theorem~\ref{thm:Gurvits} or from the Gurvits-Barnum ball defined in Theorem~\ref{Prop:Lewenstein} (numerically expressed as bounds on the Frobenius norm $||\rho_2 - \mbox{Tr}(\rho_2)\cdot\mathds{1}/D||_{F} \leq \mbox{Tr}(\rho_2)/\sqrt{D(D-1)}$) , respectively. Finally, by convexity, their sum $\rho$ is also AS.
\end{proof}

\subsection{Unifying separability criteria into a single convex program}
\label{sec:convexgeneral}
Here, we outline how to generalize the disciplined convex program from Proposition \ref{Prop:SDP} into potentially stronger sufficient conditions for separability. Although such separability criteria, as already discussed, may arise from a plethora of different techniques, it is interesting to notice that all of them can be ultimately cast as a set of inequalities, $\{f_i(\rho)\leq 0 \}$, where the functions $\{f_i\}$ are convex in the state $\rho$. Let us designate by $\mathcal{R}_i = \{\rho\geq 0\mbox{ s.t. }f_i(\rho)\leq 0 \} $ the set of (unnormalized) states detected by each criterion $i$. From the curvature of $\{f_i \}$, the sets $\{\mathcal{R}_i\}$ are convex as well. However, such sets, stemming from different conditions, may be incomparable. How can we build a new condition that includes all of them? The following theorem offers a solution to this question and generalizes Proposition~\ref{Prop:SDP}. 

\begin{theorem}
\label{thm:SDP}
Let $\rho$ be a quantum state, and $\{f_{i} \}$ a set of convex functions leading to a set of non comparable inequalities $\{f_i(\rho)\leq 0 \}$, each of them corresponding to a sufficient condition for separability. If the convex program
\begin{equation}
\label{eq:central_program}
    \begin{array}{crl}
    \min_{\{\varrho_i\geq 0\}}& 0&\\
    \mbox{s.t.}& \rho& = \sum_{i}\varrho_i\\
    &\{f_i(\varrho_i)&\leq 0\}\\
    \end{array}
\end{equation}
is feasible, then, $\rho$ is separable. 
\end{theorem}

\begin{proof}
The proof is analogous to the one of Proposition~\ref{Prop:SDP}, where the conditions for AS are introduced as generic convex constraints.   
\end{proof}

By construction, the previous theorem implicitly tests whether the state $\rho$ belongs to the convex hull of $\mathcal{R}$, where $\mathcal{R}=\cup_i\mathcal{R}_i$, is the union of the detected sets \cite{roberts_convex_1973}. 
Consequently, Theorem \ref{thm:SDP} enables the certification of states $\rho$ that are not detected by any criteria $\{f_i(\rho)\leq 0\}$ alone, but that nonetheless can be decomposed into states that are detected by such criteria. In turn, convex problems of the structure Eq.~\eqref{eq:central_program}, are not only mathematically well-posed \cite{skrzypczyk_semidefinite_2023, boyd_convex_2004}, but for some instances they may admit polynomial-time solutions. In this regard, there is a number of available solvers tailored to these problems~\cite{mosek_mosek_2024}, as well as user-friendly parsers~\cite{cvxpy_cvxpy_2024}.

\section{Multipartite absolute separability}
\label{sec:multi}
For multipartite settings, much less is known concerning (absolute) separability. We recall that a state $\rho\in \mathcal{B}(\mathbb{C}^N\otimes\mathbb{C}^M\otimes\dots\otimes \mathbb{C}^T)$ is (fully) separable if it admits a convex decomposition onto product states 
\begin{equation}
   \rho = \sum_ip_i\ketbra{e_i,f_i, \dots,g_i}, 
\end{equation}
where $\ket{e_i}\in \mathbb{C}^M$, $\ket{f_i}\in \mathbb{C}^N$, $\dots$,  $\ket{g_i}\in \mathbb{C}^T$; $\{ p_i\geq 0\}$ and $\sum_ip_i = 1$.

In complete analogy with the bipartite case, the AS set corresponds to all states that remain (fully) separable under the action of any global unitary transformation. We immediately notice that all the bipartite AS conditions discussed previously can be extended to a multipartite scenario of total dimension $D=N\cdot M\dots T$ as sufficient AS criteria with respect to all bipartitions.
Furthermore, given a state $\sigma\in\mathcal{B}((\mathbb{C}^{d})^{\otimes N})$ detected  as AS by Theorem~\eqref{Prop:Lewenstein} w.r.t.\ to all possible bipartitions, all its reduced density matrices $\Tr_{k}(\sigma)=\sigma_{N-k}$ will also be detected as AS with respect to any bipartiton by the map defined in the reduced space $(\mathbb{C}^{d})^{\otimes N-k}$. Positivity of the inverse map in the global Hilbert space leads to the following condition in the reduced space:
\begin{equation}
\label{eq:FullySeparablePartitions2}
    \sigma_{N-k}\hspace{0.01\textwidth}\geq\hspace{0.01\textwidth} \frac{d^{k}\cdot \mathds{1}}{d^{N}+2}\hspace{0.01\textwidth}\geq\hspace{0.01\textwidth}\frac{\mathds{1}}{d^{N-k}+2}.
\end{equation}
However, we recall that certifying separability with respect to all bipartitions in the state and its successive reduced matrices does not imply full separability \cite{dur_separability_1999, neven_entanglement_2018}. That is to say, certain multipartite states are entangled despite being separable with respect to its bipartitions. Paradigmatic examples of such states  can be built from  unextendible product bases (UPB)~\cite{bennett_unextendible_1999} or  are given by  multiqubit PPT entangled states in the symmetric sector~\cite{augusiak_entangled_2012, tura_four-qubit_2012}.  

The previous observation evidences the challenges to certify full separability as a much more stringent condition than biseparability. A bound on the exponent of the radius of the multipartite fully separable ball (Gurvits-Barnum) was provided ~\cite{gurvits_better_2005}. This norm-based condition represents one of the few known sufficient criteria for full absolute separability.

\begin{lemma}
\label{Lemma:Multipartite}\cite{gurvits_better_2005}
Let $\rho$ be a normalized $N$-qudit state acting on a Hilbert space of total dimension $D=d^{N}$. If
\begin{equation}
\label{eq:PurityA}
    \Tr(\rho^{2})\leq \frac{1}{D-A},
\end{equation}
with $A=2^{2-N}$, then $\rho$ is absolutely fully separable.
\end{lemma}
For the specific case of qubits, $d = 2$, this bound can be improved to $A=\frac{\beta\cdot 2^{N} }{\beta + 3^{N}}$, where $\beta = 54/17$ \cite{hildebrand_entangled_2007}. In the same spirit as Theorem~\ref{Prop:Lewenstein}, Lemma~\ref{Lemma:Multipartite} allows us to formalize the following sufficient condition for AS, stemming from the reduction-like map Eq.~\eqref{eq:ReductionMap}.

\begin{theorem}
\label{thm:alphaSep}
Let $\rho$ be an $N-$qudit state of total dimension $D=d^{N}$ and let $\Lambda_\alpha(\rho) = \mathrm{Tr}(\rho)\mathds{1} + \alpha \rho$ be a family of reduction-like maps. If 
\begin{equation}
\label{eq:boundAlphaFull}
\alpha_{-}^*\leq \alpha\leq \alpha_{+}^*,
\end{equation}
then $\sigma= \Lambda_\alpha(\rho)$ is absolutely fully separable, where 
$\alpha_{\pm}^{*}=\frac{A\pm\sqrt{A\cdot(D-1)\cdot(D-A)}}{D-A-1}$ and $A=2^{2-N}$. 
\end{theorem}
\noindent For the specific case $d=2$ we have $A=\frac{\beta\cdot 2^{N}}{\beta + 3^{N}}$, where $\beta = 54/17$.

\begin{proof}
By convexity of the AS set and linearity of the map $\Lambda_{\alpha}$ and its inverse, it is sufficient to verify the statements for pure states, $\rho = \ketbra{\Psi}$, then $\Lambda (\ketbra{\Psi})=\mathds{1} + \alpha\ketbra{\Psi}=\sigma$, whose trace and purity reads $\mathrm{Tr}(\sigma)= D + \alpha, \mathrm{Tr}(\sigma^2)= D + \alpha(\alpha +2),$  $D = d^N$. By virtue of Lemma~\ref{Lemma:Multipartite}, one can bound $\alpha$ from $A$ by solving a quadratic equation. Specifically, Eq.~\eqref{eq:PurityA} is equivalent to the positivity of a polynomial in $\alpha$ of degree two, whose extremal values $\alpha$ are given by its two roots. 
\end{proof}

By construction, Theorem~\ref{thm:alphaSep} does not detect separability beyond Lemma~\ref{Lemma:Multipartite}. However, as opposed to the condition given in  Eq.~\eqref{eq:PurityA}, our criteria depend only on the minimal or maximal eigenvalue of the density matrix and not on its full spectrum. In the following, we specify the new condition.     

\begin{corollary}
    Let $\sigma$ be an $N-$qudit state with minimal eigenvalue $\lambda_{\min}(\sigma)$ and maximal eigenvalue $\lambda_{\max}(\sigma)$. If
\begin{equation}
    \lambda_{\min}(\sigma)\geq \frac{1}{d^N+\alpha_+}\quad \text{or}\quad \lambda_{\max}(\sigma)\leq \frac{1}{d^N+\alpha_-},
\label{eq:redcondmulti}
\end{equation}
where $\alpha_\pm$ are defined in Theorem~\ref{thm:alphaSep}, then $\sigma$ is absolutely fully separable.
\end{corollary}
\begin{proof}
The proof is analogous to that of Theorem~\ref{Prop:Lewenstein}.    
\end{proof}
Surprisingly, in contrast to its bipartite counterpart, the scalable condition  Eq.~\eqref{eq:PurityA}  is not the tightest of this form. Already for $N = 2$ it is weaker than Eq.~\eqref{eq:GurvitsBi}. In fact, to date, the tight value of $A$ remains unknown. However, in~\cite{hildebrand_entangled_2007}, by exploring constructions of separable states at the border of entanglement, the authors report the upper bound $\beta = 4$. Remarkably, such bound is also of order $\mathcal{O}(1)$ in the system’s size $N$, which implies that the gap between lower and upper bounds is closed up to a constant factor.\\

The existence of multipartite states fulfilling Theorem~\ref{Prop:Lewenstein} but remaining unconsidered by Lemma~\ref{Lemma:Multipartite} introduces the notion of a subtle form of entanglement: states that are AS and AP with respect to any bipartition and with respect to any bipartition of any of the reduced matrices, but yet not fully separable. In light of the intricacies of multipartite AS, in the next section we focus on the symmetric subspace, in which the characterization of these features is more manageable, as symmetric states are either fully entangled or fully separable.

\subsection{Symmetric absolute separability and absolute PPT}
The notions of AS and AP introduced earlier may be immediately restricted to the symmetric subspace, which we shall denote here as $S(\mathcal H)$. In a generic $N-$qudit system, $\mathcal{H} = (\mathbb{C}^d)^{\otimes N}$, this subspace has dimension $\binom{N+d-1}{d-1}$. 
Specifically, we define
\begin{definition}
\label{Def:SAS}
     A symmetric state $\rho$ acting on  $S(\mathcal H)$ is symmetric absolutely separable (SAS) if and only if, given any unitary matrix acting on the symmetric subspace $U_{S}$, $\rho'=U_{S} \rho_{S} U_{S}^{\dag}$ is fully separable.
\end{definition}
\begin{definition}
\label{Def:SAP}
     A symmetric state $\rho$ acting on  $S(\mathcal H)$ is symmetric absolutely PPT (SAP) if and only if, given any unitary matrix acting on the symmetric subspace $U_{S}$, $\rho'=U_{S} \rho_{S} U_{S}^{\dag}$ is PPT with respect to any bipartition.
\end{definition}

The symmetric subspace is of particular significance for the application of the techniques presented before to certify a given convex property from the spectrum of the state. In what follows, we provide sufficient SAS and SAP criteria by studying the tight bounds of $\alpha_{\pm}$ from the inverse map condition Eq.~\eqref{Inversereduction} and the convexity criteria presented in Lemma~\ref{Lemma:GeneralAlphaPM}. For this case, it is convenient to work in the lower-dimensional Dicke basis \cite{dicke_coherence_1954} directly with symmetric matrices $\rho_{S}$ around the symmetric projector $\mathds{1}_{S}.$\\

First, we will consider the symmetric cases of $2$- and $3$-qubit systems. Even though they are simple cases, we already provide a better analytical characterization for the latter case than what is known. Following the approach from Theorem~\ref{thm:Primer}, that is, considering the inverse of the reduction map \eqref{Inversereduction} in the symmetric subspace, the  $2$-qubit bounds on $\alpha\in[-3/4,1]$ that allow to certify SAS have been already derived in the general context of separability~\cite{lewenstein_linear_2022}. Then, the different SAS criteria can be reformulated as in the following Theorem~\ref{Obs:2QubitSAS}.

\begin{theorem}
\label{Obs:2QubitSAS}
 Let $\rho_{S}\in\mathcal{S}(\mathbb{C}^{2}\otimes\mathbb{C}^{2})$ be a symmetric state of a $2-$qubit system with $\lambda_{0}\leq\lambda_{1}\leq\lambda_{2}$ its corresponding eigenvalues in non-decreasing order. If 
 \begin{equation}
    \lambda_{0} \geq \frac{1}{4} \quad\text{or}\quad \lambda_{2} \leq \frac{4}{9},
\label{eq:red2cond}
\end{equation}
then $\rho$ is SAS.

Moreover, if
\begin{equation}
\label{eq:ASSymmetric2Qubits}
7\lambda_{0}+5\lambda_{1}\geq3,
\end{equation}
then $\rho$ is SAS.
\end{theorem}
\begin{proof}
The proof of the single eigenvalue expressions is analogous to that of Theorem~\ref{Prop:Lewenstein}, and uses the derivations on the inverse of the reduction map introduced by some of us in \cite{lewenstein_linear_2022}. The linear inequality follows from Lemma~\ref{Lemma:GeneralAlphaPM} by considering $\alpha_{+}=1$, $\alpha_{-}=-3/4$ and $D=3$.
\end{proof}
However, as appreciated in Figure~\ref{fig:SimetricComparison}, even though the simplexes are tight, the CH of both of them does not attain the criterion $\sqrt{\lambda_{1}}+\sqrt{\lambda_{0}}\geq 1$ (see~\cite{champagne_spectral_2022, serrano-ensastiga_maximum_2023}), which fully characterizes SAS in the simple symmetric $2-$qubit scenario. For the case of a symmetric $3-$qubit system, we managed to provide the actual tight bounds of $\alpha_{\pm}$, which had also been obtained numerically in~\cite{serrano-ensastiga_maximum_2023}.
\begin{figure}[htbp!]
    \centering
    \begin{subfigure}[b]{0.38\textwidth}
     \raisebox{4ex}{\includegraphics[width=\textwidth]{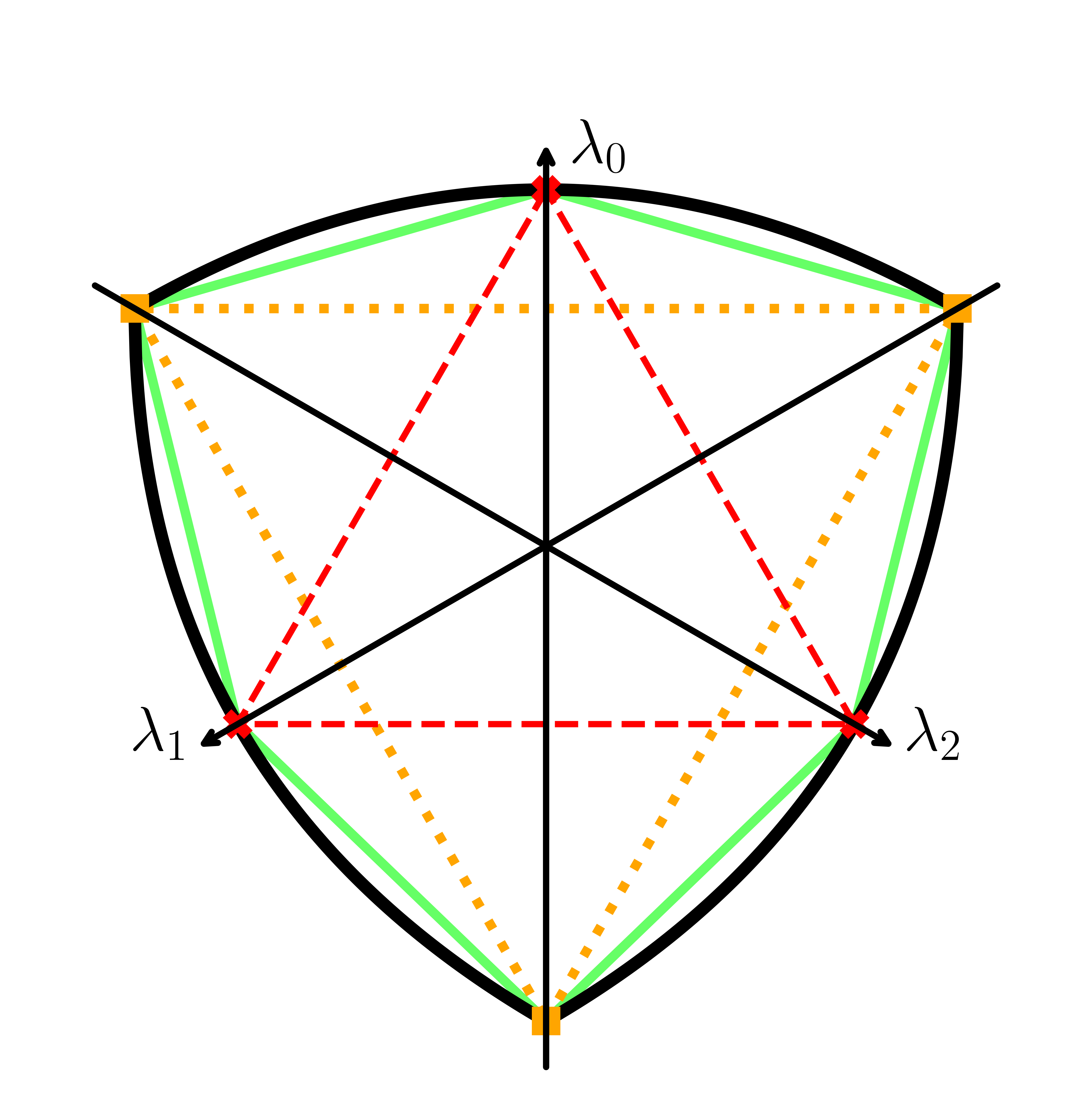}}
\end{subfigure}
\hspace{0.1\textwidth}
    \begin{subfigure}[b]{0.27\textwidth}
\includegraphics[width=\textwidth]{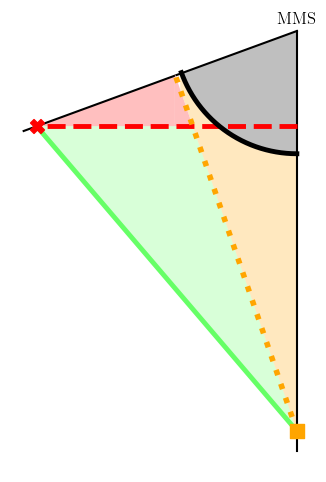}
\end{subfigure}
\vspace{-2ex}
    \caption{\textbf{Left panel:} Comparison of the states detected for a symmetric $2-$qubit system with the convex hull (green line) of the sufficient conditions through the inverse of the linear map with $\alpha=1$  (dashed red line) and $\alpha=-3/4$ (dotted orange line) with the complete characterization $\sqrt{\lambda_{0}}+\sqrt{\lambda_{1}}\geq 1$  (black solid line). \textbf{Right panel:} Comparison of the states detected for a symmetric $3-$qubit system with the convex hull (green line) of the sufficient conditions through the inverse of the linear map with $\alpha=2/3$  (dashed red line) and $\alpha=-2/3$ (dotted orange line) with the best known  analytical criterion \cite{serrano-ensastiga_absolute-separability_2024} (black sphere). The representation shows a $2D$ cut of the non-increasingly ordered region of eigenvalues $\lambda_{0}\leq\lambda_{1}\leq\lambda_{2}\leq\lambda_{3}$.} 
    \label{fig:SimetricComparison}
\end{figure}

\begin{theorem}

\label{Obs:3QubitsSAS}
 Let $\rho_{S}\in\mathcal{S}(\mathbb{C}^{2}\otimes\mathbb{C}^{2}\otimes\mathbb{C}^{2})$ be a symmetric state of a $3-$qubit system and $\lambda_{0}\leq\lambda_{1}\leq\lambda_{2}\leq\lambda_{3}$ its corresponding eigenvalues in non-decreasing order. If
 \begin{equation}
      \lambda_{0} \geq \frac{3}{10} \quad\text{or}\quad  \lambda_{3} \leq \frac{3}{14},
      \label{eq:l0l3}
 \end{equation}
 then $\rho$ is AS. 
 
Moreover, if 
\begin{equation}
    \label{eq:SAS3QUBITS}
6\lambda_{0}+6\lambda_{1}+2\lambda_{2}\geq 3,
\end{equation}
then $\rho$ is SAS.
\end{theorem}

\begin{proof}
Eq.~\eqref{eq:l0l3} is analogous to the ones in Theorem~\ref{Prop:Lewenstein}. For the positive bound of $\alpha=2/3$ see  Proposition~\ref{thm:NquditsShort} , and for $\alpha_{-}=-2/3$ see Theorem~\ref{thm:3quditsShort}. The proof of Eq.~\eqref{eq:SAS3QUBITS} follows from Lemma~\ref{Lemma:GeneralAlphaPM} by considering the announced $\alpha_{\pm}$ and $D=4$.
\end{proof}

Interestingly, for $N=3$ symmetric qubits, some SAS states, found numerically, lie outside the region given by Eq.~\eqref{eq:SAS3QUBITS}, as recently reported in~\cite{serrano-ensastiga_maximum_2023}.  Nevertheless, our analytical linear expression includes all known, to date, analytical criteria \cite{serrano-ensastiga_absolute-separability_2024, bohnet-waldraff_absolutely_2017}. In Figure~\ref{fig:SimetricComparison}, we display the known conditions that detect SAS states, together with those derived in this work. As it can be seen, the volume of SAS states detected with our conditions represents a notable improvement. \\

For higher dimensional cases, we will focus on the SAP (symmetric absolute PPT) set of states.  Let us recall that in the symmetric subspaces $S((\mathbb C^2)^{\otimes N})$ for $N=2,3$, PPT is a sufficient condition for separability, and thus SAS is equivalent to SAP in such cases. In what follows, we derive new bounds to certify SAP for the general $N-$qudit symmetric spaces with respect to any bipartitation. Very recently, some counterexamples for the presumed equivalence between the SAS and SAP sets of states have been found numerically for  $N=5,7,9$ in qubit systems \cite{louvet_equivalence_2024}. Nevertheless, the SAP set of states has only been characterized for the specific case of bipartite qudit systems, i.e., $N=2$, via a set of LMI on the eigenvalues of such systems \cite{champagne_spectral_2022}. However, the number of such inequalities grows with the local dimension $d$ as $(d(d+1)/2)!$ becoming intractable already for $d=5$. Here, we present new conditions valid for arbitrary dimensions and number of parties $d, N$, which are much more practical. 

Regarding SAP, it is also necessary to address that absolute positive partial transposition in multipartite systems is defined w.r.t.\ all the bipartitions of the state at once. For symmetric multipartite states, it can be seen (c.f.  Proposition~\ref{thm:NquditsShort} ) that the sets of SAP fulfill the inclusions $\text{SAP}_{\lfloor N/2 \rfloor: N - \lfloor N/2 \rfloor}\in \text{SAP}_{\lfloor N/2 \rfloor-1: N - \lfloor N/2 \rfloor+1}\in \cdots \in \text{SAP}_{2:N-2} \in\text{SAP}_{1:N-1}$. Thus, it is necessary to consider SAP w.r.t.\ the biggest (or most balanced) bipartition , otherwise it is always possible to find a state that is PPT under the action of any symmetric unitary in the biggest bipartition and NPT for some symmetric unitary in the smaller one. To this aim, in the following  Proposition~\ref{thm:NquditsShort} , we derive analytical bounds for $\alpha_{\pm}$ to certify SAP in $N$-qudit systems.

\begin{proposition}
\label{thm:NquditsShort} 
Let $\rho_{S}\in\mathcal{S}((\mathbb{C}^{d})^{\otimes N})$ be a symmetric state of an $N$-qudit system with $\{\lambda_{k}\leq\lambda_{k+1}\}_{k=0,\cdots, D_{S}-1}$ its corresponding eigenvalues in non-decreasing order, with $D_{S}=\binom{N+d-1}{d-1}$. If at least one of the following inequalities holds:
\begin{equation}
    \label{eq:CoroNQudits}
    \lambda_{0}\geq \frac{1}{D_{S}+2\cdot\binom{N}{\lfloor N/2 \rfloor}^{-1}}\quad\text{or}\quad \lambda_{D_{S}-1}\leq \frac{1}{D_{S}-\binom{N}{\lfloor N/2 \rfloor}^{-1}},
\end{equation}
then $\rho_{S}$ is symmetric absolutely PPT.

Moreover, if
\begin{equation}
    \label{eq:SAP}
3\binom{N}{\lfloor N/2 \rfloor}\sum_{i=0}^{\left\lceil \frac{D_{S}-2}{3} \right\rceil-1}\lambda_{i}    +\left[\binom{N}{\lfloor N/2 \rfloor}\left(D_{S}-3\left\lceil \frac{D_{S}-2}{3} \right\rceil\right)+2  \right]\lambda_{\left\lceil \frac{D_{S}-2}{3} \right\rceil}\geq\binom{N}{\lfloor N/2 \rfloor}
\end{equation}
then $\rho_{S}$ is symmetric absolutely PPT.
\end{proposition}

\begin{proof}
On the one hand, Eq.~\eqref{eq:CoroNQudits} is analogous to the ones in Theorem~\ref{Prop:Lewenstein}. It is necessary to compute the bounds of $\alpha_{\pm}$ to certify SAP following Theorem~\ref{thm:Primer}. We need to prove that $\Lambda_{\alpha}(S)=\Tr(\rho_{S})\mathds{1}_{S}+\alpha\rho_{S}\in S'$, where $S'$ is the set of PPT symmetric states. We define $k$ as the number of qudits included in the partition $A$ that is transposed, and assume that $\Tr(\rho_{S})=1$. Now, we use the linearity of the partial transpose to state that $(\Lambda_{\alpha}(\rho_{S}))^{T_{k}}=(\mathds{1}_{S})^{T_{k}}+\alpha\rho_{S}^{T_{k}}\geq 0$. Since the eigenvalues of a transposed positive semidefinite density matrix are constrained between $[-1/2,1]$ \cite{zyczkowski_volume_1998, rana_negative_2013}, the conditions become $\lambda_{min}(\mathds{1}_{S}^{T_{k}}) -\alpha_{+}/2\geq0\rightarrow \alpha_{+}\leq2\cdot\lambda_{min}(\mathds{1}_{S}^{T_{k}})$  and $\lambda_{min}(\mathds{1}_{S}^{T_{k}})-|\alpha_{-}|\geq 0\rightarrow \alpha_{-}\geq -\lambda_{min}(\mathds{1}_{S}^{T_{k}})$. Moreover, the span of any transpose matrix of the symmetric subspace of $N-$qudits is contained in $\mathcal{S}((\mathbb{C}^{d})^{\otimes k})\otimes\mathcal{S}((\mathbb{C}^{d})^{\otimes N-k})$.  From here, it is sufficient to consider first that the biggest partition given by $k=\lfloor N/2 \rfloor$ includes the remaining ones. Second, the minimal eigenvalue of $\mathds{1}_{S}^{T_{k}}$ is given by $\binom{N}{k}^{-1}$, independently of $d$. This value can be analytically proven in a case by case basis for small values of $N$. Generically, for arbitrary $d,N$, the same value is recovered numerically by expressing the partial transpose in the bipartite symmetric basis $\mathcal{S}((\mathbb{C}^{d})^{\otimes k})\otimes\mathcal{S}((\mathbb{C}^{d})^{\otimes N-k})$.

On the other hand, the proof of Eq.~\eqref{eq:SAP} follows from Lemma~\ref{Lemma:GeneralAlphaPM} by considering the anounced $\alpha_{+}=2\cdot\binom{N}{\lfloor N/2 \rfloor}^{-1}$, $\alpha_{-}=-\binom{N}{\lfloor N/2 \rfloor}^{-1}$ and $D=D_{S}=\binom{N+d-1}{d-1}$. We also considered the simplification $c=\left\lceil \frac{1}{3} \left( -3 + 2\binom{N}{\lfloor N/2 \rfloor}^{-1} +D_{S} \right) \right\rceil = \left\lceil \frac{D_{S}-2}{3} \right\rceil$. 
\end{proof}

An interesting remark is that we have numerical evidence that the upper bound of $\alpha_{+}=2\cdot\binom{N}{\lfloor N/2 \rfloor}^{-1}$ is tight for the SAP set, since increasing such value by even a $0.1\%$ leads to NPT states under certain unitary transformations. Nevertheless, the same behavior has not been found for the negative lower bound on $\alpha_{-}=\binom{N}{\lfloor N/2 \rfloor}^{-1}$, we foresee a dependence on $d$ for $\alpha_{-}$. Thus, the states that lead to a partial transpose with maximal eigenvalue $1$ do not fall in the span of the eigenvectors of $\mathds{1}_{S}$ with minimal eigenvalue \cite{zyczkowski_volume_1998, rana_negative_2013}. As an example, one can check that for $N=2,3$ qubits, the lowest bounds are given by $\alpha_{-}=-3/4, -2/3$ respectively, instead of the values $\alpha_{-}=-1/2,-1/3$ given by  Proposition~\ref{thm:NquditsShort} . For a given $N$, the negative bound becomes tighter as $d\rightarrow \infty$. Nevertheless, in a case by case basis, it is possible to study the structure of the
\begin{equation}
\label{eq:SAPMAP}
\Lambda_{\alpha}(\rho_{S})^{T_{\lfloor N/2 \rfloor:N-\lfloor N/2 \rfloor}}=\mathds{1}_{S}^{T_{\lfloor N/2 \rfloor:N-\lfloor N/2 \rfloor}}+\alpha \rho_{S}^{T_{\lfloor N/2 \rfloor: N -\lfloor N/2 \rfloor}}
\end{equation}
matrix to ensure positivity and derive better bounds of $\alpha_{-}$ as we present next. \\

\begin{theorem}
\label{thm:2quditsShort}
Let $\rho_{S}\in\mathcal{S}(\mathbb{C}^{d}\otimes\mathbb{C}^{d})$ be a symmetric state of a $2$-qudit system with $\{\lambda_{k}\leq\lambda_{k+1}\}_{k=0,\cdots, D_{S}-1}$ its corresponding eigenvalues in non-decreasing order, with $D_{S}=d\cdot(d+1)/2$. If at least one of the following inequalities holds:
\begin{equation}
    \label{eq:Coro2Qudits}
    \lambda_{0}\geq \frac{1}{D_{S}+1}\quad\text{or}\quad \lambda_{D_{S}-1}\leq \frac{1}{D_{S}-\frac{d+1}{2d}},
\end{equation}
then $\rho_{S}$ is symmetric absolutely PPT.

Moreover, if
\begin{equation}
    \label{eq:Coro2Qudits2}
     \frac{3d+1}{d+1}\cdot\sum_{i=0}^{c-1}\lambda_{i}+\left[\frac{d(d+1)}{2}-\frac{3d+1}{d+1}\cdot c +1 \right]\cdot\lambda_{c}\geq 1
\end{equation}
where  $c = \left\lceil \frac{d\cdot(d+3)\cdot(d-1)}{2\cdot(3d+1)} \right\rceil$, then $\rho_{S}$ is symmetric absolutely PPT. 
\end{theorem}
\begin{proof}
Eq.~\eqref{eq:CoroNQudits} is derived in complete analogy with Theorem~\ref{Prop:Lewenstein}. All we need is to find the range of $\alpha$ such that for any bipartite symmetric state $\rho_{\cal S}\geq 0$, $\sigma = \Lambda_\alpha(\rho) = \mathds{1}_{S} + \alpha\rho_{\cal S}$ is PPT. By convexity and covariance of the map, it is sufficient to consider pure states $\rho = \ketbra{\Psi}$ in the Schmidt form $\ket{\Psi} = \sum_{i=0}^{d-1}x_i\ket{i}_A\ket{i}_B$, $\{0\leq x_i\leq x_{i+1} \}$. Then by computing its partial transpose, we obtain\footnote{The following matrix is expressed with respect to the basis $\bigoplus_{i>j}\{\ket{ij},\ket{ji}\}\oplus \{\ket{kk}, k = 0,1,..,d-1\} $.}
\begin{equation}
\label{eq:Emofro2qudits}
\sigma^{T_A}= \bigoplus_{i>j}
\begin{pmatrix}
\frac{1}{2} & \alpha x_{i} x_{j} \\
\alpha x_{i} x_{j} & \frac{1}{2}
\end{pmatrix}\oplus\begin{pmatrix}
1 + \alpha x_{0}^{2} & \frac{1}{2} & \cdots & \frac{1}{2} \\
\frac{1}{2} & \ddots &  & \frac{1}{2} \\
\vdots &  & \ddots & \vdots \\
\frac{1}{2} & \frac{1}{2} & \cdots & 1 + \alpha x_{d-1}^{2}
\end{pmatrix}.
\end{equation}
Positivity of the first $d(d-1)/2$ blocks implies $1/2 - |\alpha|x_{d-1}x_{d-2}\geq 0$. From  $x_{d-1}x_{d-2}\leq 1/2$, it yields the range $-1\leq \alpha\leq 1$. On the other hand, by convexity, it is sufficient to set $x_{d-2} = 0, x_{d-1} = 1$ to check the positivity of the last $d\times d$ block. For $\alpha\geq 0$, the block is PSD. The minimal value of $\alpha$, $\alpha_{\min}<0$ such that the block is PSD will lead to a rank-deficient matrix. This condition is met whenever the last row $(1/2,...,1/2,1+\alpha):=\mathbf{w}$ is orthogonal to the orthogonal complement of the renaming $d-1$ rows: $(-1/d,...,-1/d, 1):=\mathbf{v}$. The equality $\mathbf{v}\cdot\mathbf{w} = 0$ is fulfilled by $\alpha_{\min} = -(d+1)/(2d)$. Thus, for $\alpha\geq -(d+1)/(2d)$, the last block is PSD as well. 

\noindent From the previous considerations, we conclude $\alpha_- = -(d+1)/(2d)$ and $\alpha_+ = 1$. Finally, Eq.~\eqref{eq:Emofro2qudits} follows from Lemma~\ref{Lemma:GeneralAlphaPM}  with the announced bounds on $\alpha$. 
\end{proof}

\begin{theorem}
\label{thm:3quditsShort}
Let $\rho_{S}\in\mathcal{S}(\mathbb{C}^{d}\otimes\mathbb{C}^{d}\otimes\mathbb{C}^{d})$ be a symmetric state of a $3$-qudit system with $\{\lambda_{k}\leq\lambda_{k+1}\}_{k=0,\cdots, D_{S}-1}$ its corresponding eigenvalues in non-decreasing order, with $D_{S}=d\cdot(d+1)\cdot(d+2)/6$. If at least one of the following inequalities holds:
\begin{equation}
    \label{eq:Coro3Qudits}
    \lambda_{0}\geq \frac{1}{D_{S} +\frac{2}{3} }\quad\text{or}\quad \lambda_{D_{S}-1}\leq \frac{1}{D_{S}-\frac{d+2}{3d}},
\end{equation}
then $\rho_{S}$ is symmetric absolutely PPT.

Moreover, if
\begin{equation}
    \label{eq:Coro3Qudits2}
\frac{3d+2}{d+2}\sum_{i=0}^{c-1}\lambda_{i}+\left[D_{S} -\frac{3d+2}{d+2}\cdot c +\frac{2}{3} \right]\cdot\lambda_{c}\geq 0
\end{equation}

where  $c = \left\lceil  \frac{3\cdot D_{S}\cdot (2 + d)-7d-2}{3\cdot(3 d+2)}  \right\rceil$, then $\rho_{S}$ is symmetric absolutely PPT. 
\end{theorem}
\begin{proof}
The proof is analogous to the one of Theorem~\ref{thm:2quditsShort}. It is sufficient to consider a pure state of the form $\ket{\Psi}=\sum_{i=0}^{d-1}x_{i}\ket{i}_A\ket{ii}_B$, $\{0\leq x_i\leq x_{i+1} \}$, and evaluate the range of $\alpha$ such that the image of the reduction map $\sigma = \Lambda_\alpha(\ketbra{\Psi}) = \mathds{1}_{S} + \alpha\ketbra{\Psi}$ is PPT. The partial transpose of $\sigma$ admits a block-diagonal form as per\footnote{The following matrix is expressed with respect to the basis $\bigoplus_{i\neq j}\{\ket{ijj},\ket{jii}\}\bigoplus_k \{\ket{kkk},\{\ket{llk}_{l=0,...,d-1; l\neq k},\ket{l'kl'}_{l=d-1,...,0; l'\neq k} \} \} \} \bigoplus_{p\neq q\neq r} \{\ket{pqr},\ket{rpq} \}$. } 
\begin{equation}
\label{eq:Emofro3qudits}
\sigma^{T_A}= \bigoplus_{i\neq j}
\begin{pmatrix}
\frac{1}{3} & \alpha x_{i} x_{j} \\
\alpha x_{i} x_{j} & \frac{1}{3}
\end{pmatrix} 
\bigoplus_{k}
\underbrace{
\begin{pmatrix}
1+\alpha x_{k}^{2} & \frac{1}{3} &\cdots & \cdots  & \cdots&\frac{1}{3} \\
\frac{1}{3} & \frac{1}{3} & \frac{1}{6} & \cdots & \frac{1}{6} & \frac{1}{3} \\
\vdots & \frac{1}{6} & \frac{1}{3} & \cdots  & \frac{1}{3} & \frac{1}{6} \\
 \vdots & \vdots & & \ddots & & \vdots \\
 \vdots & \frac{1}{6} & \frac{1}{3} & \cdots &\frac{1}{3} & \frac{1}{6} \\
 \frac{1}{3} & \frac{1}{3} & \frac{1}{6} & \cdots & \frac{1}{6} & \frac{1}{3}
\end{pmatrix}}_{:=M_k}\bigoplus_{p\neq q\neq r} \begin{pmatrix}
    1/6 & 1/6 \\
      1/6 & 1/6 
\end{pmatrix}.
\end{equation}

Positivity of the first $d(d-1)/2$ blocks implies $1/3 - |\alpha|x_{d-1}x_{d-2}\geq 0$. From  $x_{d-1}x_{d-2}\leq 1/2$, it yields the range $-2/3\leq \alpha\leq 2/3$. On the other hand, for the next $d$ blocks, it is sufficient to check the positivity of $M_{d-1}$ with $x_{d-2} = 0, x_{d-1} = 1$. For $\alpha\geq 0$, $M_{d-1}\geq  0$. The minimal value of $\alpha$, $\alpha_{\min}<0$ such that $M_{d-1}\geq 0 $ will decrease its rank. For this condition, we can construct a $d$-dimensional set of linearly independent vectors as the orthogonal complement of the matrix $M_{d-1}$ once the first row has been deleted. The first $(d-2)$ ones are of the form $(0, ..., 0, 1, 0, ..., 0, -1, 0, ..., 0)$ where the $1$ and $-1$ are set in the positions $1+\kappa$ and $-\kappa$ of the vector respectively, with $\kappa = 1,2,..,d-2$. The ones restraining the  value  of alpha are of the form $(-d/2,1,...,1,0,...,0):=\mathbf{v}_{1}$ and $(-d/2,1,...,1,0,1,0,...,0):=\mathbf{v}_{2}$, both of them containing $d/2$ ones. It can be seen trivially that all the vectors of the first form are orthogonal to the first row $(1+\alpha, 1/3, ..., 1/3):=\mathbf{w}$ of $M$. The equality $\mathbf{v}_{i}\cdot\mathbf{w} = 0$ is fulfilled by $\alpha_{\min} = -(d+2)/(3d)$. Thus, for $\alpha\geq -(d+2)/(3d)$, the last block is PSD as well. 

\noindent From the previous considerations, we conclude $\alpha_{-} = -(d+2)/(3d)$ and $\alpha_+ = 2/3$. Finally, Eq.~\eqref{eq:Coro3Qudits2} follows from Lemma~\ref{Lemma:GeneralAlphaPM}  with the announced bounds on $\alpha$. 
\end{proof}

Remarkably, we present analytical tight bounds of the parameter $\alpha_{\pm}$ for arbitrary dimensions $d$ and $N=2,3$.  Contrarily to the bounds for $\alpha_{-}$ in  Proposition~\ref{thm:NquditsShort} , as expected from the analytical derivations presented in Theorems~\ref{thm:2quditsShort} and \ref{thm:3quditsShort}, it is possible to find NPT states arbitrarily close to the given extreme values of $\alpha_{-}$.  In general, these bounds are algebraic numbers defined by the roots of characteristic polynomials of high degree. For $N>3$, the structure of the matrices is more involved (see e.g. recent contribution~\cite{romero-palleja_multipartite_2024}). The systematic extension of our results to arbitrary number of partitions $N$ warrants further investigation and it will be examined elsewhere.

\section{Conclusions}
\label{sec:conclus}
 In this work we have presented a new method to tackle the problem of absolute separability and absolute PPT based on inverting non-completely positive maps. The AS (AP) sets are defined as those formed by the states whose separability (positivity under partial transposition) is preserved under any global unitary transformation. These sets are convex and contain the maximally mixed state as well as states sufficiently close to it. Due to this invariant condition, AS and AP states are completely determined from their spectrum. 

Our method provides novel simple analytical \textit{sufficient} conditions for AS in arbitrary dimensions based on the knowledge of very few or even one eigenvalue of the spectrum. These analytical conditions correspond to convex simplexes that contain extremal points of the AS set, improving presently known conditions for AS. Moreover, we have demonstrated that AS conditions can be cast as convex functions of the form $f(\rho)\leq 0$. This fact allows us to unify all known AS conditions into a single and stronger one using convex programming. Such a technique can be generalized to any convex set to improve its characterization, which is normally a very hard task. Our method also yields new conditions on the problem of absolute separability in the multipartite scenario, as well as in the case of the symmetric subspace. Finally, from the same methodology, we provide various new criteria on the twin problem of AP for multipartite symmetric systems. 

While some protocols for teleportation \cite{garg_teleportation_2021} or error correction \cite{koczor_dominant_2021}  aim at the evaluation of the maximal eigenvalue of the state, our criteria mostly rely on the knowledge of the lowest one, motivating the study of experimental setups that lower-bound the spectrum of a given state. Moreover, our results of absolute separability in the multipartite scenario highlights the present lack of results in characterizing the entanglement properties of highly mixed systems, prompting deeper studies on the topic.

Finally, it is worth noting that similar techniques as the one we report here can be used to study the invariance of other convex sets of states under global unitary transformations. For instance, our method can be applied to further analyze absolute non-violation of steering \cite{ss_bhattacharya_absolute_2018} or Bell-CHSH \cite{ganguly_bell-chsh_2018} inequalities. It can also be used to introduce and characterize other invariant convex sets, such as absolute Schmidt Number \cite{abellanet2025private}. Lastly, the possibility of exploiting these absolute sets in a resource theory remains an open question.\\

\vspace{0.3cm}
\noindent \textbf{Acknowledgements}.
-- We thank A. Diebra and J. Romero-Pallejà for insightful comments and  acknowledge  fruitful discussions with E. Serrrano-Ensátiga, J. Ahiable, K.N.B. Teja, J. Tura and A. Winter.

J.A.-V. acknowledges financial support from Ministerio de Ciencia e Innovación of the Spanish Goverment FPU23/02761. We acknowledge financial support from Spanish MICIN (projects: PID2022: 141283NB100; 139099NBI00) with the support of FEDER funds, the Spanish Goverment with funding from European Union NextGenerationEU (PRTR-C17.I1),
the Generalitat de Catalunya, the Ministry for Digital Transformation and of Civil Service of the Spanish Government through the QUANTUM ENIA project -Quantum Spain Project- through the Recovery, Transformation and Resilience Plan NextGeneration EU within the framework of the Digital Spain 2026 Agenda. 
 We also acknowledge support from: Europea Research Council AdG NOQIA; MCIN/AEI (PGC2018-0910.13039/501100011033, CEX2019-000910-S/10.13039/501100011033), EU-QUANTera Projects: MAQS PCI2019-111828-2, DYNAMITE PCI2022-132919, ExTRaQT-PCI2022-132965.
 Support from Fundació Cellex; Fundació Mir-Puig; Generalitat de Catalunya (European Social Fund FEDER and CERCA program, AGAUR Grant No. 2021 SGR 01452, QuantumCAT \ U16-011424, co-funded by ERDF Operational Program of Catalonia 2014-2020); Barcelona Supercomputing Center MareNostrum (FI-2023-3-0024) funded by the European Union are acknowledged. 
 This project has received funding from EU Horizon Europe Program ORIZON-CL4-2022-QUANTUM-02-SGA PASQuanS2.1, 101113690; FET-OPEN OPTOlogic, Grant No 899794, NeQSTGrant-01080086. ICFO Internal “QuantumGaudi" project.
 Views and opinions expressed are, however, those of the authors only and do not necessarily reflect those of the European Union or the European Research Council. Neither the European Union nor the granting authority can be held responsible for them.
 G. R.-M. acknowledges funding from the European Innovation Council accelerator grant COMFTQUA, no. 190183782.
 
\bibliographystyle{quantum}
\bibliography{BiblioTFMQ}
\newpage
\appendix
\section{Numerical approach to Lemma~\ref{Lemma:GeneralAlphaPM}}
\label{AppendixDerivationSimplex}
In this appendix, we present a numerical method to compute the linear inequalities defining the convex hull of a set of extremal points of the absolute separability (AS) set. This approach is particularly useful when the assumptions underlying Lemma~\ref{Lemma:GeneralAlphaPM} are not met, when the symmetry between the primal and dual simplices is broken, as it might be the case for other sets of extremal points \cite{song_extreme_2024, serrano-ensastiga_maximum_2023}.

Working in the eigenvalue $\{\lambda_{i}\}$ space, normalized by the trace constraint $\sum_{i}\lambda_{i}=1$, allows the AS criteria to be expressed in terms of convex polytopes in $(D-1)$-dimensional space. A basis orthogonal to $(1,...,1)$ provides a convenient coordinate system. The polytope can be represented in its $H$-representation as a set of inequalities $A\lambda\leq b$, where $A$ and $b$ are derived from the vertices of the convex hull. In cases with permutation invariance, $A$ consists of all permutations of a single inequality vector \cite{fukuda_combinatorial_1994}.

The analitical computation of $A,b$ is difficult to extend to high dimensions. Instead, numerical tools such as \textit{pypoman} \cite{caron_pypoman_2024} can be used to compute $A$, $b$. Since in the normalized space $A,b$ might depend on the basis chosen, it is advisable to convert them into expressions of the whole space of eigenvalues. Since the representation might depend on the chosen basis, the inequalities are lifted back to the full $D$-dimensional eigenvalue space to preserve generality.

To identify a minimal set of necessary inequalities for a fixed eigenvalue ordering, we consider removing a candidate inequality $\mathbf{a}^{T}\boldsymbol{\lambda}\leq \beta$  and augmenting the remaining system with the normalization and the specific ordering ($\{\lambda_{k}\leq \lambda_{k+1}\}_{k=0,1,...,D-2}$) constraints. Denoting the new system as $A'\boldsymbol{\lambda}\leq\mathbf{b}'$, we solve the linear program
\begin{equation}
\label{eq:LP}
\text{max}\{ \mathbf{a}^{T}\boldsymbol{\lambda}| A'\boldsymbol{\lambda}\leq\mathbf{b}'\}=\gamma.
\end{equation}

If $\gamma\leq\beta$, the removed inequality is redundant. Repeating this process allows reduction to the minimal necessary set.\\

For the case of Theorem~\ref{Prop:Lewenstein}, the vertices $\mathbf{e}_{i}$ are permutations of $(1-(D-1)/(D+\alpha), 1/(D+\alpha),...,1/(D+\alpha))$. The resulting inequalities in the full eigenvalue space take a form such that $\mathbf{b}$ is a constant vector, and $A$ is a matrix of permuted vectors $\mathbf{c}=(p, \cdots, p, r, -q, \cdots, -q)$. The coefficients $p,q,r$ and their multiplicities depend only on $D,\alpha_{\pm}$, ensuring permutation invariance. The number of inequalties scales as $D\times\binom{D}{\max{m(p), m(q)}}\leq D!$. Moreover, the reduction on the number of inequalities can be done until just one expression is needed a part from the ordering and the normalization expressions. 

\section{Schur convexity}
\label{Appendix:Ansatz}
In this section, we revisit key concepts of \textit{majorization}, as outlined in~\cite{marshall_inequalities_2011}, along with the references within the text. 

To begin with, if a vector $\mathbf{x}$ is majorized by $\mathbf{y}$, that is, $\mathbf{x}\prec\mathbf{y}$ then $\mathbf{x}$ is less spread than $\mathbf{y}$. In the context of the space of eigenvalues of quantum density matrices, we refer as $\mathbf{x}$ being more mixed than $\mathbf{y}$. Thus, majorization theory is closely related with absolute separability and PPT. Formally, 

\begin{definition}
\label{DefinitionMajorization}
Consider two $D$-dimensional vectors $\mathbf{x}$ and $\mathbf{y}$ ordered in a non-decreasing way ($x_0\leq\cdots\leq x_{D-1}$), we say that $\mathbf{x}\prec\mathbf{y}$ if
\begin{equation}
\label{eq:Majorization1}
    \sum_{i=0}^{D-1}x_i = \sum_{i=0}^{D-1}y_i
\end{equation}
\begin{equation}
\label{eq:Majorization2}
    \sum_{i=0}^{k}x_i \geq \sum_{i=0}^{k} y_i, \quad k=0,\cdots, D-2
\end{equation}
\end{definition}
\noindent We also present two lemmas, that will be of great utility.
\begin{lemma}
\label{lemma:simplex}
A given vector $\mathbf{x}$ is majorized by another vector $\mathbf{y}$ if and only if $\mathbf{x}$ is in the convex hull of the $D!$ permutations of $\mathbf{y}$.
\end{lemma}
From the previous lemma, one can immediately infer that any point in the space of eigenvalues of a density matrix inside the simplexes given by Eq.~\eqref{Inversereduction} is majorized by any of the vertices of the simplex, which are all different permutations of a single one. The same notion can be applied to any of the extreme points that might be derived for AS and AP for better characterization of both sets.

\begin{lemma}
\label{lemma:Convexity}\cite{marshall_inequalities_2011, day_rearrangement_1972}
If $\mathbf{a}\prec \mathbf{b}$, $\mathbf{u}\prec \mathbf{v}$ and if $\mathbf{b}$ and $\mathbf{v}$ are similarly ordered, then $\mathbf{a}+\mathbf{u}\prec\mathbf{b}+\mathbf{v}$.
\end{lemma}
Since the majorization condition  requires  the fulfillement of the normalization Eq.~\eqref{eq:Majorization1}, the previous lemma implies that any point inside of the convex hull of the two simplexes is majorized by some points on the border of the convex hull (obtained as a linear combination of the vertex of the simplexes).

Next, we introduce the concept of Schur-convex functions as those non-increasing under majorization: 
\begin{definition}
A function $f:\mathbb{R}^{D}\rightarrow\mathbb{R}$ is defined Schur-convex if and only if for all $\mathbf{x},\mathbf{y}\in\mathbb{R}^{D}$ such that $\mathbf{x}$ is majorized by $\mathbf{y}$ one has that $f(\mathbf{x})\leq f(\mathbf{y})$.
\end{definition}
The next lemma provides a necessary and sufficient conditions for a function to be Schur-convex. 
\begin{lemma}
\label{thm:lemmaSchur}
A function $f(\mathbf{x})$ is Schur-convex if it fulfills the following conditions:
\begin{itemize}
    \item $f(\mathbf{x})$ must be symmetric under permutations of the arguments of $\mathbf{x}$
    \item $(x_{i}-x_{j})\cdot\left(\frac{\partial f}{\partial x_{i}}- \frac{\partial f}{\partial x_{j}}\right)\geq 0$ for all $\mathbf{x}\in\mathbb{R}^{D}$.
\end{itemize}
\end{lemma}

We now verify that any vector contained in the convex hull of the two simplexes fulfills the proposed condition in Lemma~\ref{Lemma:GeneralAlphaPM}. First, given the eigenvalues of a density matrix in non-decreasing order $\{\lambda_{k}\leq\lambda_{k+1}\}_{k=0,\cdots,D-2}$ we define
\begin{equation}
f(\boldsymbol{\lambda}) = 1-\left(1-\frac{\alpha_{+}}{\alpha_{-}} \right)\cdot \sum_{i=0}^{c-1}\lambda_{i}-\left[D+\alpha_{+}+c\cdot\left(\frac{\alpha_{+}}{\alpha_{-}}-1\right)\right]\cdot \lambda_{c},
\end{equation}
from Eq.\eqref{eq:GeneralAlphaPM} such that our condition is expressed as $f(\boldsymbol{\lambda})\leq 0$ (cf. Section\ref{sec:convexgeneral}). Here, the argument of the function is any vector of eigenvalues, which is not necessarily sorted. This function is clearly symmetric under permutations, as it accounts for the ordering of each component of the vector independently of its position.  Furthermore, it can be verified from Lemma~\ref{thm:lemmaSchur} that for natural values of $D \in \mathbb{N}$, $\alpha_{+}>0$ and $\alpha_{-}<0$; the function $f(\boldsymbol{\lambda})$ is Schur-convex. On the other hand, we recall that the vertices of the AS simplexes are given by all the permutations of the vectors  $\boldsymbol{\lambda}_{\pm} = \left(\frac{1}{D+\alpha_{\pm}}, \cdots, \frac{1}{D+\alpha_{\pm}}, 1-\frac{D-1}{D+\alpha_{\pm}}\right)$, in the space of eigenvalues. Due to the convexity of the AS set and Schur-convexity of the function, we check the inequality $f(\boldsymbol{\lambda})\leq 0$ on the extremals $\boldsymbol{\lambda}_{\pm }$. In this regard, we confirm the tight value  $f(\boldsymbol{\lambda}_{\pm}) = 0$. Thus, it follows from Schur-convexity that all the vectors contained in the convex hull of the points generated by the permutations of $\boldsymbol{\lambda}_{\pm}$  fulfill the inequality. Finally, the procedure outlined here can be extended to Eqs.~\eqref{eq:TwoSimplex}, \eqref{eq:ASSymmetric2Qubits}, \eqref{eq:SAS3QUBITS}, \eqref{eq:SAP}, \eqref{eq:Coro2Qudits2} and \eqref{eq:Coro3Qudits2} of the main text.

\end{document}